\documentclass{article}

\usepackage[active]{srcltx}
\usepackage{a4wide}
\usepackage[utf8]{inputenc}
\usepackage[T1]{fontenc}
\usepackage{comment}
\usepackage{a4,fullpage}
\usepackage{epsfig}
\usepackage{amssymb,amsmath,amsthm}
\usepackage{xspace}
\usepackage{longtable}
\usepackage{tikz}
\usetikzlibrary{decorations.pathmorphing,patterns,calc}

\theoremstyle{plain}
\newtheorem{theorem}{Theorem}
\newtheorem{lemma}[theorem]{Lemma}
\newtheorem{fact}[theorem]{Fact}

\newtheorem{corollary}[theorem]{Corollary}

\theoremstyle{definition}

\theoremstyle{remark}

\theoremstyle{remark}

\begin{document}

\title{\bf Graphs with large chromatic number induce $3k$-cycles\thanks{Work partially 
done during the 2014 Barbados Workshop on Structural Graph Theory}}

\author{Marthe Bonamy\thanks{Universit\'e Montpellier 2 - CNRS, LIRMM
    161 rue Ada,
    34392 Montpellier, France.
    Email: \texttt{bonamy@lirmm.fr} Partially supported by the ANR Project \textsc{EGOS} 
    under \textsc{Contract ANR-12-JS02-002-01}}
  \and Pierre Charbit\thanks{LIAFA - Universit\'e Paris Diderot - 
Paris 7 - Case 7014 - F-75205 Paris Cedex 13. Email: \texttt{pierre.charbit@liafa.univ-paris-diderot.fr}}
  \and St\'ephan Thomass\'e \thanks{LIP, UMR 5668, ENS Lyon - CNRS - 
UCBL - INRIA, Universit\'e de Lyon, France.
   Email: \texttt{stephan.thomasse@ens-lyon.fr}
  Partially 
supported by the ANR Project \textsc{Stint} under \textsc{Contract ANR-13-BS02-0007}
}
}

\maketitle

\begin{abstract}
Answering a question of Kalai and Meshulam, we prove that graphs without
induced cycles of length $3k$ have bounded chromatic number. This implies the very 
first case of a much broader question asserting that every graph with 
large chromatic number induces a graph $H$ such that the sum of the Betti
numbers of the independence complex of $H$ is also large.

\end{abstract}

\section{Introduction}

One of the simplest results in graph theory is
that any graph without odd cycles is bipartite. To 
the contrary, one of the hardest results (the 
celebrated strong perfect graph theorem \cite{crst}) asserts that forbidding 
odd induced cycles of length at least five and their complements
produces graphs where the chromatic number equals the clique number.
Parity of cycles is indeed strongly related to coloring, and many
questions are still open in this field, mainly because the following
problem is still not really understood: What is the general picture
of graphs having large chromatic number and small clique number?

In particular, if one considers a triangle-free graph $G$ with
large chromatic number, the general feeling
is that $G$ must be 'complex' in the sense that it contains 
all sorts of induced substructures. Unfortunately, the vast majority 
of natural questions one can ask in this direction remains unsolved. It is a shame that 
even proving in this case the existence of an induced cycle of length at least
say 10 (to be provoking) is still wide open.

This paper investigates induced cycles of length 
0 modulo 3. To avoid confusion with parity, it 
was convenient for us to speak about the {\it trinity} 
of an integer, which is its residue modulo 3. By extension, the
{\it trinity} of a path or a cycle is the trinity of its length. 
Finally, a {\it trinity graph} is
a graph which does not induce trinity 0 cycles. Observe that trinity graphs 
have clique number at most 2. The goal of this paper 
is to show that trinity graphs have bounded
chromatic number.

Studying the trinity of induced cycles in a graph may appear to be an exotic 
goal, so let us give some motivation for this. A teasing result (distantly related to our
problem) where trinity of cycles plays a crucial role is the following: Assume that 
$G$ is a (4-regular) graph on $n$ vertices whose set of edges is partitioned into 
two cycle-factors $F_1$ and $F_2$. Assume moreover that $F_1$
is a union of triangles. Then if all cycles in $F_2$ have trinity 
0 or 2, there exists a stable set $S$ of size $n/3$, hence
meeting every triangle in $F_1$ exactly once. The existence 
of $S$ is not constructive, and for instance no polynomial 
algorithm to find $S$ is known even when $F_2$ simply consists 
of a disjoint union of $C_5$.

So what is so special about trinity 0 and 2 cycles? The answer,  
given by Aharoni and Haxell \cite{ah}, can be found 
in the hypergraph consisting of all 
stable sets of $G$, or equivalently, the {\it independence complex} 
${\cal I}(G)$ of $G$. To give a very light intuition, observe that 
${\cal I}(C_4)$ consists of two disjoint intervals, ${\cal I}(C_5)$ 
consists of one cycle, and ${\cal I}(C_6)$ consists of two 
(plain) triangles, attached by three edges. Here the important parameter 
$\eta ({\cal I}(G))$ is the dimension of the smallest 'hole' in ${\cal I}(G)$
(more precisely the first non trivial homology group), where 
$\eta =0$ when the complex is not connected, $\eta =1$ when the complex is connected but
not simply connected, etc. The crucial observation is that $\eta ({\cal I}(C_k))+1$
is at least $k/3$ when the trinity of $k$ is 0 or 2, and strictly 
less than $k/3$ when the trinity of $k$ is 1. In particular, the topological connectivity
of the independent complex of cycles of trinity 1 is not large enough to 
give an independent set of representative
(hence a stable set of size $n/3$) as in the
Aharoni-Haxell theorem \cite{ah}.

This is part of the answer: trinity of cycles is a key-parameter when 
considering the independence complex of a graph. But can we go further than 
simply proving the existence of (linear) stable sets? Do some properties of
${\cal I}(G)$ directly give bounded chromatic number? More precisely, 
if we ask that ${\cal I}(G)$, and all ${\cal I}(H)$ for $H$ induced 
in $G$ are 'simple enough', is it true that $G$ has bounded chromatic number?

This idea was developed by Gil Kalai and Roy Meshulam, and the parameter they 
proposed is to consider, for a given graph $H$, the sum $bn(H)$ of all reduced Betti numbers 
of ${\cal I}(H)$ (i.e. the sum of the number of independent holes in each dimenion, or 
more precisely the sum of the ranks of all homology groups). They conjectured
that if $G$ has large chromatic number, then one of its induced subgraphs
$H$ has large $bn(H)$. Observe that large cliques have in particular large 
parameter $bn$, and that, if true, this conjecture would imply the 
existence of a ``complex`` induced subgraph (at least with respect to
some parameter which is typically large for complete graphs). This would be 
a milestone considering how poor is our knowledge of chromatic number.

Going back to our toy-examples, one can notice that ${\cal I}(C_6)$ has two non-equivalent 
(1-dimensional) holes, while ${\cal I}(C_4)$ and ${\cal I}(C_5)$ first 
non trivial homology groups have rank 1. This remark generalizes as follows:
the (unique) non trivial homology group of ${\cal I}(C_n)$ has rank 1
if $n$ has trinity 1 or 2, otherwise it has rank 2. 
Therefore, a graph only inducing graphs $H$ with $bn(H)\leq 1$ does not
have induced cycles of length $3k$. 
Hence, if one wants to show the first 
nontrivial case of the Kalai-Meshulam
conjecture, i.e. that there exists an induced subgraph $H$
with $bn(H)>1$, it would suffice in particular to show that every graph 
with large chromatic number has an induced $3k$-cycle. 
This is the goal of this paper.

We do not make any attempt to provide an 
explicit bound, in order to avoid tedious computations 
which would hide the quite simple general idea of the proof. To achieve this, we will only 
characterize the numbers we use in our proof as {\it bounded} or {\it large}, 
where large will be taken in the sense of 'arbitrarily large'. Hence in our
proof, the output of any unbounded increasing function applied 
to 'large' remains 'large'. Note that a parameter which is not arbitrarily 
large is bounded, so every parameter we use in this paper 
is either large or bounded. 

We did not tried to figure out the bound we can derive from our method,
but we are pretty confident that its decimal expansion should easily fit in one line. However, we
feel a little bit sorry since the threshold seems to be 4, as suggested 
by the existing bound when excluding trinity 0 cycles as subgraphs.
Indeed Chudnovsky, Plumettaz, Scott and Seymour \cite{css} proved that every 
graph with chromatic number at least four has a (not necessarily induced) trinity 0 cycle.

Let $X$ and $Y$ be some disjoint sets of vertices of some 
graph $G$. We say that $X$ and $Y$ are \textit{independent} if there 
is no edge between them, otherwise we say that $X$ {\it sees} $Y$. 
We say that $X$ has {\it private neighbors} in $Y$
if for every $x\in X$ there is a neighbor $y$ of $x$ in 
$Y$ such that $y$ only sees $x$ in $X$.
We say that $X$ {\it dominates} $Y$ if every vertex of 
$Y$ has a neighbor in $X$. Moreover $X$ {\it minimally dominates} $Y$ if 
$X$ {dominates} $Y$ and no $S \subsetneq X$ {dominates} $Y$. Observe 
that if $X$ {minimally dominates} $Y$, then $X$ has {private neighbors} in $Y$.

Let $G$ be a connected graph and $r$ be some vertex of $G$.
We define the {\it iterated neighborhoods} of $r$ as $N_0(r)=\{r\}$, $N_1(r)=N(r), \ldots$, 
and drop the parameter when there is no ambiguity. We say that $N_\ell$ is the {\it level} $\ell$. 
We refer to $N_0\cup N_1\cup \ldots \cup  N_\ell$ as the \textit{levels 
at most $\ell$}. If we are discussing $N_\ell$, then the \emph{upper level} 
is $N_{\ell-1}$, and similarly the \emph{lower level} is $N_{\ell+1}$.
If $x,y$ are vertices at equal distance to $r$ (or equivalently,
belonging to the same level $N_{\ell}$),
we denote by $U_{xy}$ any induced $xy$-path with internal 
vertices belonging to levels $N_{i}$ with $i<\ell$, and such that 
$U_{xy}\cap N_i$ has at most 2 vertices for all $i$. To form 
$U_{xy}$, one can take a shortest $xr$-path $P$, a shortest 
$yr$-path $Q$, and extract a shortest $xy$-path inside $V(P\cup Q)$.
Here $U_{xy}$ stands for ``up-path'' from $x$ to $y$. We also 
denote by $u_{xy}$ the trinity of the path $U_{xy}$, and more generally,
when we consider some paths $P_{xy}$ or $T_{xy}$, their respective
trinities are denoted by $p_{xy}$ or $t_{xy}$.

\section{Trinity Changing Paths}

A {\it trinity changing path} (TCP) is a sequence of graphs such 
that we can go through each graph with
two induced paths with different trinities. Formally, a TCP of {\it order} $k$ is
obtained by considering a sequence of pairwise disjoint and non-adjacent 
graphs $G_1,\dots ,G_k$ and vertices $x_1,y_1,\dots ,x_k,y_k$, where $x_i,y_i$ 
are in $G_i$, and then identifying every vertex $y_i$ with $x_{i+1}$ when
$i=0,\dots ,k-1$. 
Moreover, in each $G_i$, where $1\leq i\leq k$, there exist
two induced $x_iy_i$-paths with different trinities. In particular $x_i\neq y_i$ and 
$x_iy_i$ is not an edge. The vertex $x_1$ is called the {\it origin} of the TCP
and every $G_i$ is a {\it block}.

\begin{theorem}\label{th:TCP}
Every vertex $r$ in a connected trinity graph $G$ with large chromatic number is the origin 
of a large order TCP.
\end{theorem}

\begin{proof}
We just show how to grow the first block from $r$.
Let $\ell$ be such that the chromatic
number of $G[N_\ell]$ is large. Let $G'$ be some connected component of
$G[N_\ell]$ with large chromatic number. Inside $G'$,
we choose a vertex 1, and grow an induced path 12345678 in $G'$ in such a way 
that there exists a subset of vertices $X$ in $G'$ nonadjacent to 
1234567, the vertex 8 has a unique neighbor 9 in $X$, and $X$ induces 
a connected graph with large chromatic number. We call $X'$ the graph 
induced by 12345678 union $X$. To construct $X'$,
we apply the classical argument of Gy\'arf\'as~\cite{Gyar87} showing that every vertex
in a large chromatic triangle-free connected graph is the origin 
of a large induced path: Start by vertex 1 and remove $N(1)$ from
$G'$. Since $N(1)$ is a stable set, there is a connected component $G''$ of large chromatic number
in $G'\setminus N(1)$. Then pick $2\in N(1)$ which sees $G''$. Now
start from 2 inside $2\cup G''$ and proceed as done previously from 1 in $G'$.
Iterating this process gives the path 12345678 and the set $X$ which 
is used as a 'storage' of large chromatic number. 

We denote by $S$ a dominating set of $X'$ in $N_{\ell-1}$.

\begin{fact}\label{fact:Sstable}
If $S$ is not a stable set, we can grow our first block.
\end{fact}

\begin{proof}
Let $xy$ be an edge of $S$. We consider $X''$ to be a connected component
of $X'\setminus (N(x)\cup N(y))$ with large chromatic number.
Let $z$ be a vertex of $X'\setminus X''$ with at least one neighbor in $X''$. 
By construction $z$ is a neighbor of $x$ or $y$, but not of both since that 
would induce a triangle: assume w.l.o.g. that $xz$ is an edge. Let $P_x$ and $P_y$
be two shortest paths from $r$ to $x$ and from $r$ to $y$. Our first block consists
of $P_x\cup P_y\cup z$. Note that the two paths $P_xz$ and $P_yxz$ have 
different trinities, and that $P_yxz$ is indeed induced since there is no triangle. 
\end{proof}

From now on, we assume that $S$ is a stable set. Let us say that two vertices 
$i$ and $i+2$ in $1234567$ are {\it clean} if $\{i,i+2\}$ has private neighbors in $S$.

\begin{fact}\label{fact:clean}
There exist clean vertices $i$ and $i+2$.
\end{fact}

\begin{proof}
First observe that if $1,3,5$ have a common neighbor $x$ in $S$, then $2,4$ 
must be clean, otherwise a common neighbor $y$
of $2,4$ would give the $6$-cycle $1x54y2$. We concude similarly if $3,5,7$ 
have a common neighbor in $S$. If $1,3$ are not clean, there is a common neighbor $x$ 
of $1,3$. If $5,7$ are not clean, there is a common neighbor $y$ of $5,7$. 
Since $x$ does not see $5$ and $y$ does not see $3$, we have our clean vertices $3,5$.
\end{proof}

Now let us consider our clean vertices $i,i+2$ and their 
respective private neighbors $x$ and $y$ in $S$. 
If one of $x$ and $y$ has a neighbor in $X$, we consider
a connected component $C$ of $X\setminus (N(x)\cup N(y))$ with large
chromatic number together with a vertex 
$z$ of $X\cap (N(x)\cup N(y))$ having at least one neighbor in $C$. Without loss 
of generality, $z$ is joined to $x$. Observe that $z$ cannot be adjacent to $y$, 
for $xzy(i+2)(i+1)i$ would be a $C_6$. 
Our first block consists of $P_x\cup P_y\cup \{z,i,i+1,i+2\}$. Note 
that the two paths $P_xz$ and $P_y(i+2)(i+1)ixz$ have 
different trinities, and that $P_y(i+2)(i+1)ixz$ is indeed induced to avoid a $C_6$.

We now assume that $x$ and $y$ have no neighbors in $X$. We denote by 
$z$ the neighbor of $x$ or $y$ with maximum index in 12345678.
We moreover denote by $P_x$ and $P_y$
two shortest paths from $r$ to $x$ and from $r$ to $y$.

\begin{enumerate}
 
\item If $z=i+2$, our first block consists of $P_x\cup P_y\cup \{i,i+1,i+2\}$. Note 
that the two paths $P_xi(i+1)(i+2)$ and $P_y(i+2)$ have different trinities.

\item If $z=i+3$, then we have the edge $x(i+3)$, and 
our first block consists of $P_x\cup P_y\cup \{i+2,i+3\}$ with
the two paths $P_x(i+3)$ and $P_y(i+2)(i+3)$ having
different trinities.

\item If $z>i+3$, without loss of generality $xw$ is an edge,
and our first block consists of $P_x\cup P_y\cup \{i,i+1,i+2,w\}$
with the two paths $P_xz$ and $P_y(i+2)(i+1)ixz$ having
different trinities.
\end{enumerate}

We can now build the next blocks starting from $z$. 
\end{proof}

Our way of using TCP's will involve inserting inside the blocks
of our TCP some fixed subgraphs like for instance $C_5$. To do so, 
we need a slightly stronger version of the previous theorem, 
which is easily obtained through a similar proof.

\begin{theorem} \label{incrust}
Let $H$ be some fixed graph and $G$ be some connected trinity graph with large chromatic number.
Then one of the following holds:
\begin{enumerate}
\item Every vertex $r$ of $G$ is the origin 
of a large order TCP in which every block induces a copy of $H$.
\item There exists an $H$-free induced subgraph of $G$ with 
large chromatic number. 
\end{enumerate}

\end{theorem}

\begin{proof}
The proof is exactly the same as that of the previous theorem,
we just have to append to every block a copy of $H$.

We just show how to grow the first block.
Let $r$ be a vertex of the graph. Let $\ell$ be such that $G[N_\ell(r)]$
has large chromatic number. Let $G'$ be some connected component of
$G[N_\ell]$ with large chromatic number. 
If $G'$ is $H$-free, then we are done. If not, we find a copy of
$H$ (denoted by $H$). Consider a shortest path $Q$ from 
$r$ to $H$, and denote by $Z$ the union of $Q$ and $H$.
This is the subgraph that we will add to our first 
block, once we will have constructed it. Observe that
$N(Z)\cap G'$ has bounded chromatic number (namely at most 
$|H|+1$, since the neighborhood of any set $S$ in $G$ has chromatic number at most $|S|$ and
there is only one vertex in $Q \cap N_{\ell-1}$), 
and thus we can consider a connected component $G''$ of
$G'\setminus N(Z)$ with large chromatic number.

From this point, we grow our first block in the same way as 
in the previous theorem. The only difference is that we
add $Z$ to it.                                                      
\end{proof}

Let us conclude this section on TCP by observing the key-property
of these structures: if some vertex $x$ sees some blocks 
$B_i$ and $B_k$ of a TCP $T$, then $x$ sees every pair of blocks 
$B_j\cup B_{j+1}$ where $i<j<j+1<k$. Indeed, we could otherwise form 
an induced cycle using $x$ and the TCP which would traverse $B_j$ and $B_{j+1}$,
leaving then all possible choices of trinities including trinity 0.

\section{Rich sets and TCPs}

We now use TCPs to find many disjoint independent induced
paths between vertices. Let $X$ and $Y$ be disjoint sets in a graph $G$. 
We say that $X$ is {\it rich in $Y$} if we can find in $G[Y]$ a large number of 
pairwise independent subgraphs $Y_1, Y_2, \ldots, Y_p$ (called also {\it blocks}) such that every vertex 
in $X$ has at least one neighbor in each $Y_i$. 
Observe that if $u,v$ are vertices of $X$, then there exist many (in the sense 
of 'a large number of')
independent induced $uv$-paths. In particular, if $uv$
is not an edge, at most one of these paths has trinity 0, and 
the others have either all trinity $1$, or all trinity $2$. We then 
say that $uv$ has {\it type} $1$ or $2$, respectively. If $uv$
is an edge, we define the {\it type} of $uv$ to be 1.

Rich sets are very useful to find trinity 0 cycles.
It could be tempting to conclude that if we have a stable set of six rich vertices in $Y$,
then since every pair of vertices has type 1 or 2, one can find by Ramsey a triangle
$uvw$ of the same type, and hence a trinity 0 cycle. This is however not the case,
since a triangle $uvw$ of type 2 could be obtained for instance if $u,v,w$ always have 
the same neighbors in $Y$, or a triangle $uvw$ of type 1 is such that $u$ 
has a neighbor on every $vw$-path with internal vertices in $Y$. 

Let us formalize this notion. Let $u,v, w$ be vertices in the rich set $X$. For any $i$, we say that
$u$ \textit{cuts $vw$ in $Y_i$} if every induced $vw$-path in $\{v,w\}\cup Y_i$ 
contains a neighbor of $u$. More generally, we say that $u$ \textit{cuts $vw$} 
if $u$ cuts $vw$ in all but a bounded number of $Y_i$'s.

Let $S$ be a subset of $X$. We say that some $Y_i$ is a \textit{typical block} 
(with respect to $S$) if for every triple $u,v,w \in S$, if $u$ cuts $vw$ then $u$ cuts 
$vw$ in $Y_i$. Observe that such a block exists for every bounded subset $S$.

Let us see now how we can construct rich sets using TCPs. The following 
lemma is the key of the next section where it will be repeatedly used
to form our shadows and antishadows.

\begin{lemma}
Let $G$ be a connected trinity graph with large chromatic number
and $r$ be a vertex of $G$. If $N_{\ell}(r)$ has large chromatic 
number, then there exists a set $X$ in $N_{\ell -1}(r)$ which
is rich in $N_{\ell}(r)$ and such that $X$ dominates a TCP in $N_{\ell}(r)$.
\end{lemma}

\begin{proof}
Start with any vertex $u$ in $N_{\ell-1}$ which sees a connected component 
$C$ of $N_{\ell}$ with large chromatic number. Consider a 
connected component $C'$ of $C\setminus N(u)$ with large chromatic number.
Let $z$ be a vertex of $C\cap N(u)$ which sees $C'$.
Now start growing a TCP $T$ from $z$ in $C'$. 
Observe that if some vertex $x$ of $N_{\ell-1}$ has a neighbor $y$ in some block $B_i$ of 
$T$, then we can close a cycle using an $U_{xu}$-path
and some induced $zy$-path $P$ on the TCP. Since $P$ 
has two trinity choices when traversing each block, the only way to 
avoid a trinity 0 cycle is that $x$ itself sees many blocks. Precisely,
$x$ must see every pair of blocks $B_j\cup B_{j+1}$ where $1<j<i$.
In particular, if $T$ has $2k$ blocks for some large value $k$,
the set $X$ of vertices in $N_{\ell-1}$ which 
see some $B_i$ with $k\leq i\leq 2k$ is rich in $T$, hence in $N_{\ell}(r)$.
Moreover, $X$ dominates the subTCP of $T$ consisting of all blocks 
with index at least $k$.
\end{proof}

Let us now turn to two useful facts that we will widely use in our proofs.
The first one directly follows from the fact that two rich vertices have type 1 or 2.

\begin{fact}\label{fact:richtype}
 If two vertices $x, y \in N_\ell$ are rich in a set $B$, then there 
cannot be both an induced $xy$-path of trinity $1$ and one of trinity $2$ with all 
their internal vertices independent of $B$.
\end{fact}

The second fact is crucial. We will often prove that some 
vertex cuts some pair. Here we show that there cannot be too
many cuts.

\begin{fact}\label{fact:asymmetricut}
Let $x,y,z$ be a stable set which is rich in $B$. If $x$ cuts $yz$ 
and $y$ cuts $xz$ then $x$ and $y$ have common neighbors in a large 
number of blocks of $B$. In particular $xy$ has type 2.
\end{fact}

\begin{proof}
By definition, $x$ cuts $yz$ in all but a bounded number of blocks, 
and similarly with $y$ and $xz$. Consequently $x$ cuts $yz$ 
and $y$ cuts $xz$ in all but a bounded 
number of blocks. Consider a typical
block $B_i$ in which $x$ cuts $yz$ and $y$ 
cuts $xz$. Consider a shortest $yz$-path $P$ with internal vertices in 
$B_i$. Note that $P$ has an internal vertex $v$ joined to $x$. If $v$
is not the neighbor of $y$ in $P$, there is an $xz$-path 
which is not seen by $y$, a contradiction. Consequently, for every such block $B_i$, $x$ 
and $y$ have a common neighbor. Therefore, they have a large number of common neighbors.
\end{proof}

\section{Shadows and antishadows}

Let $G$ be a connected trinity graph with large chromatic number and $r$ be a vertex of $G$.
A {\it shadow} is a set of vertices $X$ in some level $N_{\ell +1}(r)$ such that $G[X]$ 
has large chromatic number, and $X$ is dominated by a stable set $S$ included in 
$N_{\ell}$.

Let $X$ be a set of vertices in $G$ and $\ell$
be the largest value for which $X\cap N_{\ell}$ is non empty. 
We say that $X$ is {\it freely closable} if for every 
$u\in X\cap N_{\ell}$, there is an induced $ru$-path whose internal 
vertices are independent of $X\cap N_i$ for all $i<\ell$.
In other words, we can connect $r$ with $X\cap N_{\ell}$ without 
touching the vertices of $X$ living in the upper levels.

An {\it antishadow} is a collection of subsets $X_0,X_1,X_2,X_3,X_4$
included respectively in $N_{\ell}$, $N_{\ell+1}$, $N_{\ell+2}$, $N_{\ell+3}$, $N_{\ell+4}$
with the following properties:

\begin{itemize}
 \item $X_0$ induces a large chromatic number subgraph,
 \item $X_1$ induces a stable set,
 \item $X_{i+1}$ dominates $X_i$ for every $i=0,1,2,3$,
 \item every $X_i$ is rich in $X_{i+1}$, for $i=1,2,3$,
 \item $X_0\cup X_1\cup X_2\cup X_3\cup X_4$ is freely closable.
\end{itemize}

The terminology antishadow refers to the fact that, in a shadow, the 
set of large chromatic number is dominated by a stable set that lives 
in the upper level, while here 
$X_0$ is dominated by a stable set living in the lower level.

\begin{theorem}\label{shadow}
If $G$ is a connected trinity graph with large chromatic number and 
$r$ is a vertex of $G$, there exists a shadow or an antishadow.
\end{theorem}

\begin{proof}
We assume that there is no shadow, and construct an antishadow.
Observe that for every set $S\subseteq N_{k-1}(r)$ with bounded chromatic 
number, the subgraph induced by $N(S)\cap N_k(r)$ does not have large chromatic 
number. Indeed $S$ would have a stable set $S'$ such that $N(S')\cap N_k(r)$ 
has large chromatic number, hence $S'$ would be a shadow.

Let us start for this in some level $N_k$ with 
large chromatic number, and consider a component $C_k$
of $N_k$ with large chromatic number. Pick some vertex $v$ in $C_k$,
and fix some shortest $rv$-path $P$. This is this particular $P$ which will 
certify that our antishadow will be freely closable. We now build 
our first subset $X_4$ inside $N_{k-1}$. 

We perform for this a TCP extraction, which we will also 
perform in the upper levels. We just describe here the one we do 
in $C_k$.
We start at $v$, and denote by $R$ the set of vertices of 
$C_k$ which are neighbors of $v$, or have a common neighbor with $v$ in $N_{k-1}$.
In other words the vertices of $R$ which are not neighbors 
of $v$ are neighbors of the set $N(v)\cap N_{k-1}$.
Since $N(v)\cap N_{k-1}$ is a stable set and $G$ has no shadow,
the chromatic number of $R$ is not large. Hence $C'_k:=C_k\setminus R$
has large chromatic number. We now consider a shortest (induced) path $Q$ in $C_k$
from $v$ to $C'_k$. Let us say that $Q$ ends at $v'\in C'_k$. Starting at 
$v'$, we can now grow a TCP inside $C'_k$, but we will only grow a partial one. In other words, 
we grow a large TCP $T$ inside $C'_k$ starting from $v'$
and ending in $w'$ in such a way that the TCP can still be extended
from $w'$. More precisely, $w'$ belongs to a connected 
graph $C''_k$ of large chromatic number, included in $C'_k$, such that there is 
no edge between $T\setminus w'$ and $C''_k$. Now we consider the second half $T'$ of $T$, i.e. the blocks
of $T$ with index at least half the order of $T$. Set $X':=N(T')\cap N_{k-1}$.
Assume for contradiction that $X'$ does not have large chromatic number. Since $G$ has no shadow,
$N(X')\cap N_{k}$ is not large, in particular $C''_k\setminus N(X')$
has large chromatic number, and thus is non empty. Now we consider a vertex 
$w''$ in $C''_k\setminus N(X')$. To get a contradiction, consider a neighbor 
$a$ of $v$ in $N_{k-1}$, a neighbor $b$ of $w''$ in $N_{k-1}$, a path $Q'$
from $w'$ to $w''$ in $C''_k\setminus N(X')$, and finally a path $U_{ab}$. The union of $Q\cup Q'\cup U_{ab}$
is a path closing the TCP $T$. However since neither $a$ nor $b$ has 
neighbors in $T'$, we can construct a trinity 0 cycle.

Therefore $X'$ has large chromatic number, and every vertex of 
$X'$ sees all the pairs of consecutive  blocks
of $T\setminus T'$. So we have constructed a subset of $N_{k-1}$ rich in $T\setminus T'$, with large 
chromatic number. To avoid interference with $P$, we set our first set $X_4$
to be $X'\setminus N(P)$. Observe that $N(P)$ has chromatic number at most 6, hence 
$X_4$ still has large chromatic number. 

Thus, we can also perform a TCP-extraction inside $X_4$, to form $X_3$, $X_2$,
and $X_1$.

We now perform our usual TCP-extraction from $X_1$, and define
our $X'$ to be the set of neighbors of the half TCP $T'$ of $X_1$.
The chromatic number of $X'$ is still large.
The trick is now to observe that the chromatic number of 
$T'$ itself is bounded. So $T'$ can be partitioned into a bounded 
number of stable sets. One of these stable sets $S$ dominates a subset
$X''$ of $X'$ with large chromatic number. We now set $X_1$ to be equal 
to $S$, and $X_0:=X''\setminus N(P)$ to avoid interference with $P$.
\end{proof}

\section{Excluding $C_5$.}

We now turn to the proof of our main result. The technique we will 
use is now ready. We need however a last ingredient which is very
classical when considering graph classes: we have to find a 
special graph $H$ which splits the difficulty of the problem. 
More precisely, we want $H$ such that both excluding $H$ and 
knowing that $H$ is a subgraph lower the difficulty of the problem.
In our case, we do not have to look very far, we just have to 
consider $C_5$. 

We first conclude when $C_5$ is an induced subgraph.

\begin{lemma}\label{lem:dominatingC5}
Let $C$ be a 5-cycle in $N_\ell$ with a minimal dominating stable set $S$ 
either in $N_{\ell-1}$ or in $N_{\ell+1}$.
If $S$ is rich in a TCP $T$ in its lower level 
and $T$ is independent of $C$, then $|S|=3$
and every vertex of $S$ has exactly two neighbors in $C$.
\end{lemma}
\begin{proof}
Note that, to avoid triangles, no vertex in $S$ can have three neighbors in $C$ or more. 
It follows that $|S|\geq 3$. 
Let $C=\{1,2,3,4,5\}$. If there is no element of $S$ with exactly one neighbor in $C$, 
then the result holds by minimality of $S$. We consider the case where there is at least one.

First assume that there are two elements $x,y \in S$ and $i \in C$, 
such that $x$ is only adjacent to $i$ in $C$, and $y$ is only adjacent 
to $i+2$. Then there is an induced
path of trinity $1$ in $C\cup \{x,y\}$ between $x$ and $y$, as well as an induced path of 
trinity $2$. This contradicts 
Fact~\ref{fact:richtype}.

Assume now that there are at least two elements in $S$, each with exactly one neighbor in $C$, 
and we are not in the previous case. Then they are exactly two, and they are adjacent to 
consecutive vertices $i,i+1$ of $C$. Therefore $i+3$ is dominated by a vertex 
of $S$ with two neighbors in $C$, but the other dominated 
vertex should be $i$ or $i+1$, a contradiction.

Consequently, there is exactly one element $x \in S$ with only one neighbor in $C$. 
Let $y$ and $z$ be two other elements of $S$. Then, w.l.o.g. $x$ is adjacent to $1$, 
$y$ to $2,4$ and $z$ to $3,5$. Observe that $xy$, $xz$ and $yz$ have type 1 since 
they can be pairwise connected by a path of length 4 in $C$. 
Only in this step of the proof do cases of $S \subseteq N_{\ell+1}$ or $N_{\ell-1}$ 
differ. In the former case, we consider $P_{xy}$ to be the $xy$-path of length $4$ 
with internal vertices in $C$, and in the latter we consider it to be $U_{xy}$. Note 
that since, in the latter case, we can form an induced cycle of trinity $u_{xy}$ by 
considering $u_{xy}$ closed with the $xy$-path of trinity $0$ with internal vertices 
in $C$, we know that $u_{xy}=1$. Now the same considerations can be made about the two 
other pairs. We merge again the two cases of $S \subseteq N_{\ell+1}$ or $N_{\ell-1}$: 
regardless of that, we have for each pair a path $P_{xy}$ of trinity $1$ with internal 
vertices in the upper levels of $S$. In particular, every vertex $a \in N_{\ell-2}$ is 
adjacent to at most one element of $\{x,y,z\}$. As a consequence, $z$ is adjacent to no 
vertex in $U_{xy}$, and symmetrically. By Fact~\ref{fact:asymmetricut}, at most one of 
$x,y,z$ cuts the two other. W.l.o.g. we can assume that $z$ does not cut $xy$ and $x$ 
does not cut $yz$. Then we can find two independent $xy$ and $yz$-paths in the TCP, 
both of trinity $1$. By combining them with $P_{xz}$, we obtain a cycle of trinity 
$0$, a contradiction.
\end{proof}

\begin{lemma}\label{lem:antishadowC5}
 If $G$ is a trinity graph inducing an antishadow $X_0,X_1,\ldots,X_4$, 
then $X_0$ induces no $5$-cycle.
\end{lemma}
\begin{proof}
 Assume for contradiction that $X_0$ induces a $5$-cycle $C$ with vertex 
set $V(C)=\{1,2,3,4,5\}$. Let $X$ be a minimal dominating set of $C$ in $X_1$. 
By Lemma~\ref{lem:dominatingC5}, we have w.l.o.g. $X=\{x,y,z\}$, respectively 
 joined to $\{1,3\}$, $\{2,4\}$ and $\{2,5\}$.

If $y,z$ do not have a common neighbor in $X_2$, then for any $y'$ and 
$z'$ respective neighbors of $y$ and $z$ in $X_2$, the path $z'z54yy'$ has trinity 2 and 
$z'z2yy'$ has trinity 1, a contradiction to Fact~\ref{fact:richtype} applied to $y',z'$.

So $y$ and $z$ have a common neighbor $z'$ in $X_2$. 
Observe that to avoid $6$-cycles, $x$ and $y$ (resp. $z$) do not have a common neighbor in $X_2$.
Let $x'$ be a neighbor of $x$ in $X_2$. 

If $z',x'$ do not have a common neighbor in $X_3$, then for any $z'',x''$ respective 
neighbors of $z',x'$ in $X_3$, the paths $z''z'z51xx'x''$
and $z''z'z543xx'x''$ have respective trinities 1 and 2,
a contradiction to Fact~\ref{fact:richtype}.

Thus $z',x'$ have a common neighbor $z''$ in $X_3$. Now consider 
a neighbor $z'''$ of $z''$ in $X_4$. Since the antishadow is 
freely closable, there exists an $rz'''$-path $P$ whose vertices
(save $z'''$) are independent from $C,x,y,z,x',z',z''$. Let $5'$
be a neighbor of $5$ in  $N_{\ell -1}$. We can consider the lowest 
common root of $5'$ and $P$ in the BFS, which may be an element of $P$.
That way, we get an induced $5'z'''$-path $Q$ whose internal vertices
are independent from $C,x,y,z,x',z',z''$. Note that $z'''z''z'z5Q$ and
$z'''x''x'x15Q$ are induced cycles. If $z'''x''x'x345Q$
is also induced, then we have three induced cycles all of different trinities, 
a contradiction. Consequently, $5'$ and $3$ are adjacent: in particular, $5'$ is 
not adjacent to $2$. Then we consider the induced cycle $z'''z''z'y215Q$ and 
reach the same contradiction.
\end{proof}

An {\it extended $C_5$} is a graph on nine vertices consisting of a $C_5=\{1,2,3,4,5\}$
with three additional pairwise non-adjacent vertices $x,y,z$ respectively linked to
 $\{1,3\}$, $\{2,4\}$ and $\{2,5\}$, and finally an additional vertex $w$ linked
 only to $y$ and $z$ (see Figure~\ref{fig:extendedC5})
 
 \begin{figure}[t]
\centering
\begin{tikzpicture}[scale=0.7,auto]
\tikzstyle{blacknode}=[draw,circle,fill=black,minimum size=5pt,inner sep=0pt]
\tikzstyle{greennode}=[draw,circle,fill=green,minimum size=3pt,inner sep=0pt]
\tikzstyle{greennodeb}=[draw,circle,fill=green,minimum size=3pt,inner sep=0pt]
\tikzstyle{spring}=[thick,decorate,decoration={zigzag,pre length=0.05cm,post length=0.05cm,segment length=6}];
\tikzstyle{springb}=[thick,decorate,decoration={zigzag,pre length=0.3cm,post length=0.3cm,segment length=6}];
\foreach \i in {1,...,3}
{\draw (0,0)
++(-360*\i/5+90+2*72:2) node[blacknode] (a\i) [label=90:$\i$] {};}
\foreach \i in {4,...,5}
{\draw (0,0)
++(-360*\i/5+90+2*72:2) node[blacknode] (a\i) [label=-90:$\i$] {};}
\draw (0,0)
++(-360*5/5+90+2*72:2) node[blacknode] (a0) {};

\foreach \i in {1,...,5}
{\pgfmathtruncatemacro{\j}{mod(\i,5)+1}
\draw (a\i) -- (a\j);}

\draw (0,0)
++(-360*2/5+90+2*72:4) node[blacknode] (b2) [label=90:$x$] {};  
\draw (0,0)
++(-360*1/5+90+2*72:4) node[blacknode] (b1) [label=90:$z$] {};  
\draw (0,0)
++(-360*3/5+90+2*72:4) node[blacknode] (b3) [label=90:$y$] {};  
\foreach \i in {1,...,3}
{\pgfmathtruncatemacro{\j}{mod(\i+1,5)}
\pgfmathtruncatemacro{\k}{mod(\i-1,5)}
\draw (b\i) -- (a\j);
\draw (b\i) -- (a\k);}

\draw (0,-3) node[blacknode] (c) [label=-90:$w$] {};  
\draw (b1) edge [bend right] node {} (c);
\draw (b3) edge [bend left] node {} (c);
  
\end{tikzpicture}
\caption{An extended $C_5$.}
\label{fig:extendedC5}
\end{figure}
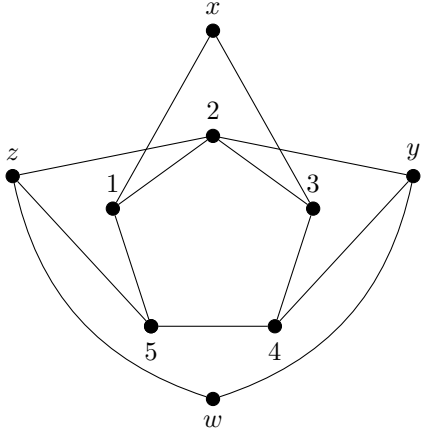

\begin{theorem} \label{C5}
If the chromatic number of trinity graphs not inducing  
$C_5$ is bounded, then  the chromatic number of trinity graphs not inducing 
extended $C_5$ is also bounded.
\end{theorem}

\begin{proof}
We suppose here that trinity-free $C_{5}$-free graphs have bounded 
chromatic number. Assume for contradiction that $G$ is a connected 
trinity-free extended $C_5$-free graph with large chromatic number. 
Let us fix some root $r$. 
According to Theorem~\ref{shadow} we can split into two cases.

Assume first that $G$ has an antishadow $X_0,X_1,\dots, X_4$
where $X_0$ belongs to $N_{\ell}$. 
Since $X_0$ has large chromatic number (and is obviously trinity-free), 
by assumption we can find a 5-cycle $C$ inside $X_0$, a contradiction to Lemma~\ref{lem:antishadowC5}.

Assume then that $G$ has a shadow. In particular one can find a large 
chromatic number subset $B$ in some $N_{\ell+1}$ which is dominated 
by some stable set $S$ in $N_{\ell}$. Since $G$ does not contain a $C_5$-free graph 
with large chromatic number, by Theorem~\ref{incrust}, we can
find a TCP inside $B$ such that each block contains a $C_5$. In particular,
we can find a 5-cycle $C$ in $B$ minimally dominated by a stable set $X$ in $S$
such that $X$ is rich. Assume that $V(C)=\{1,2,3,4,5\}$. By Lemma~\ref{lem:dominatingC5}, 
we have $|X|=3$ and every vertex in $X$ has exactly two neighbors in $C$. W.l.o.g., $X=\{x,y,z\}$, 
respectively joined to $\{1,3\}$, $\{2,4\}$ and $\{2,5\}$.

Let $x'$, $y'$ and $z'$ be respectively some neighbors of $x$, $y$ and 
$z$ in $N_{\ell-1}$. Note that $x$ is not joined to $y'$ to 
avoid the cycle $xy'y451$. Similarly, $x$ is not joined to $z'$, nor $x'$ to $y$ or $x'$ to $z$.
Note that if $y$ and $z$ have a common neighbor $w'$ in $N_{\ell-1}$, 
then $C\cup \{x,y,z,w'\}$ forms an extended $C_5$.
Therefore the only edges between $x,y,z$ and $x',y',z'$ are $xx',yy',zz'$.

In particular, a path $U_{xy}$ is nonadjacent to $z$. Since none of the two induced
cycles $U_{xy}y21x$ and $U_{xy}y451x$ have trinity 0, the trinity of $U_{xy}$
must be $1$. Hence $z51xU_{xy}y$ is an induced path with trinity 1
from $z$ to $y$.
Note that $y2z$ is an induced path of trinity 2, hence the type of $yz$ is 2.
Since $z51xU_{xy}y$ has trinity 1, $x$ must cut $yz$.
There exists a path $U_{yz}$ which is nonadjacent to $x$. Note that the trinity of $U_{yz}$
must be 2 since both $U_{yz}z2y$ and $U_{yz}z54y$ are induced cycles. We then have the trinity 
2 path $x34yU_{yz}z$. Moreover, since $x345z$ is an 
induced path of trinity 1, the type of $xz$ is 1.
Thus $y$ must cut $xz$, a contradiction to Fact~\ref{fact:asymmetricut} since $xy$ has type 1.
 \end{proof}

A {\it doubly extended $C_5$} is any graph on thirteen vertices constructed as follows:
We start with an extended $C_5$ on vertex set $\{1,2,3,4,5,x,y,z,w\}$.
We add three new pairwise non 
adjacent vertices $x',y',z'$ respectively linked to
 $\{3,5\}$, $\{1,4\}$ and $\{2,4\}$, and finally a new vertex $w'$ linked
 to $y'$ and $z'$ (see Figure~\ref{fig:doublyextendedC5}). 
 To the edges described so far, one can add any new edge
between $x',y',z',w'$ and $x,y,z,w$ provided that this does not create trinity
0 induced cycles.

 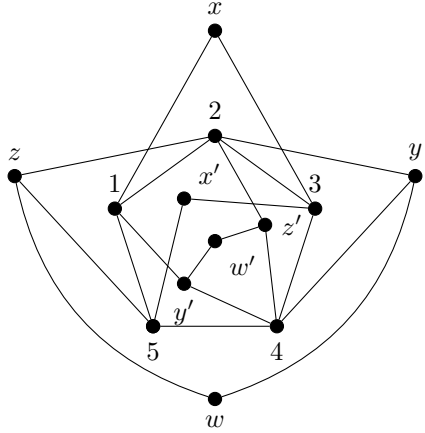
\begin{figure}[t]
\centering
\begin{tikzpicture}[scale=0.7,auto]
\tikzstyle{blacknode}=[draw,circle,fill=black,minimum size=5pt,inner sep=0pt]
\tikzstyle{greennode}=[draw,circle,fill=green,minimum size=3pt,inner sep=0pt]
\tikzstyle{greennodeb}=[draw,circle,fill=green,minimum size=3pt,inner sep=0pt]
\tikzstyle{spring}=[thick,decorate,decoration={zigzag,pre length=0.05cm,post length=0.05cm,segment length=6}];
\tikzstyle{springb}=[thick,decorate,decoration={zigzag,pre length=0.3cm,post length=0.3cm,segment length=6}];
\foreach \i in {1,...,3}
{\draw (0,0)
++(-360*\i/5+90+2*72:2) node[blacknode] (a\i) [label=90:$\i$] {};}
\foreach \i in {4,...,5}
{\draw (0,0)
++(-360*\i/5+90+2*72:2) node[blacknode] (a\i) [label=-90:$\i$] {};}
\draw (0,0)
++(-360*5/5+90+2*72:2) node[blacknode] (a0) {};

\foreach \i in {1,...,5}
{\pgfmathtruncatemacro{\j}{mod(\i,5)+1}
\draw (a\i) -- (a\j);}

\draw (0,0)
++(-360*2/5+90+2*72:4) node[blacknode] (b2) [label=90:$x$] {};  
\draw (0,0)
++(-360*1/5+90+2*72:4) node[blacknode] (b1) [label=90:$z$] {};  
\draw (0,0)
++(-360*3/5+90+2*72:4) node[blacknode] (b3) [label=90:$y$] {};  
\foreach \i in {1,...,3}
{\pgfmathtruncatemacro{\j}{mod(\i+1,5)}
\pgfmathtruncatemacro{\k}{mod(\i-1,5)}
\draw (b\i) -- (a\j);
\draw (b\i) -- (a\k);}

\draw (0,-3) node[blacknode] (c) [label=-90:$w$] {};  
\draw (b1) edge [bend right] node {} (c);
\draw (b3) edge [bend left] node {} (c);

\draw (0,0)
++(-360*4/5+90+2*72:-1) node[blacknode] (d4) [label=45:$x'$] {};  
\draw (0,0)
++(-360*5/5+90+2*72:1) node[blacknode] (d5) [label=-90:$y'$] {};  
\draw (0,0)
++(-360*3/5+90+2*72:1) node[blacknode] (d3) [label=right:$z'$] {};  
\foreach \i in {3,...,5}
{\pgfmathtruncatemacro{\j}{mod(\i+1,5)}
\pgfmathtruncatemacro{\k}{mod(\i-1,5)}
\draw (d\i) -- (a\j);
\draw (d\i) -- (a\k);}

\draw (0,0)
++(-360*4/5+90+2*72:0) node[blacknode] (e) [label=-45:$w'$] {};  
\draw (d3) edge  node {} (e);
\draw (d5) edge  node {} (e);
\end{tikzpicture}
\caption{A doubly extended $C_5$.}
\label{fig:doublyextendedC5}
\end{figure}

\begin{theorem} \label{extendedC5}
If the chromatic number of trinity graphs not inducing  
an extended $C_5$ is bounded, then the chromatic number of trinity graphs not inducing 
a doubly extended $C_5$ is also bounded.
\end{theorem}

\begin{proof}
We assume that trinity-free extended $C_5$-free graphs have bounded chromatic number, let $G$ be a connected trinity free graph with large chromatic number, and fix some root $r$. We are going to prove that it must contain an induced doubly extended $C_{5}$ or get a contradiction. According to Theorem~\ref{shadow} we have a shadow or an antishadow.

As in the proof of Theorem~\ref{C5}, if $G$ contains an antishadow, then the set $X_{0}$ must contain an extended $C_{5}$, hence an induced $C_{5}$, which contradicts Lemma~\ref{lem:antishadowC5}, so we can assume that there is a shadow. In particular one can find a large chromatic number subset $B$ in some $N_{\ell+1}$ which is dominated 
by some stable set $S$ in $N_{\ell}$. Since trinity-free extended $C_5$-free graphs have bounded chromatic number, by Theorem~\ref{incrust}, we can
find a TCP inside $B$ such that each block contains an extended $C_5$. In particular,
we can find an extended 5-cycle $C$ in $B$ with $V(C)=\{1,2,3,4,5,x,y,z,w\}$ which is 
minimally dominated by a stable set $X'$ in $S$
such that $X'$ is rich. Let $X \subseteq X'$ a minimal dominating set of 
$\{1,2,3,4,5\}$. By Lemma~\ref{lem:dominatingC5}, $X$ has three vertices, each joined 
to exactly two vertices of $\{1,2,3,4,5\}$. Up to symmetry, here are
the different cases:

\begin{enumerate}
 \item $X$ has three vertices $x',y',z'$ respectively 
 joined to $\{1,3\}$, $\{2,4\}$ and $\{2,5\}$.

We claim that there is some $a' \in X'$ that dominates 
both $1$ and $y$. Otherwise, we have in particular that $x'$ 
is not adjacent to $y$. Let $a'$ be a neighbor of $y$ in $X'$. 
Then any minimal dominating set of the $5$-cycle $y2154$ 
contained in $\{x',a',y',z'\}$ must contain $x'$, a contradiction to Lemma~\ref{lem:dominatingC5}.

Similarly, there is a vertex $b'$ of $X'$
dominating both $z$ and 3. Note that in particular, $b',a'$ do not see 
4, 5. Let $c'$ be a neighbor of $w$ in $X'$.
Note that $c'$ is not adjacent to $5$ ($c'wya'15$ 
would induce a $6$-cycle) nor adjacent to $4$ ($c'wzb'34$ would induce a $6$-cycle).
In particular, neither $z'$ nor $y'$ are adjacent to $w$.
Observe now that $54ywz$ is minimally dominated by some set $S \subseteq \{y',z',a',b',c'\}$. 
Note that $c' \in S$ since $w$ is adjacent to no other element of $\{y',z',a',b',c'\}$. 
Note that $c'$ has only one neighbor in $54ywz$, a contradiction to Lemma~\ref{lem:dominatingC5}.

 \item $X$ has three vertices $x',y',z'$ respectively 
 joined to $\{2,4\}$, $\{3,5\}$ and $\{1,3\}$.

We claim that there is a choice of $y'$ and $z'$ such that both are adjacent to $y$. First of all, note that if one sees $y$, then the other one must also see $y$ (since $y2154$ is then minimally 
dominated by $x',y',z'$, and Lemma~\ref{lem:dominatingC5} applies). We assume for contradiction that 
none of $y',z'$ can be chosen adjacent to $y$.
Let $a'$ be a neighbor of $y$ in $X'$. Let $S \subseteq \{x',y',z',a'\}$ be a minimal 
dominating set of the $5$-cycle $y2154$. Then $a'$ cannot be adjacent to both $1$ and $5$. 
Assume w.l.o.g. that it is not adjacent to $1$. To avoid a contradiction with Lemma~\ref{lem:dominatingC5}, 
we must have $a'$ adjacent to $5$ and $z'$ to $y$. Now, by a similar argument on $\{x',z',a'\}$ 
and the $5$-cycle $12345$, we obtain that $a'$ is adjacent to $2$ or $3$. Because the previous 
case did not apply, $a'$ must be adjacent to $3$, a contradiction to the choice of $y'$. 

Let $b'$ be some vertex dominating $x$ (where we choose
$b'=x'$ if indeed $x'$ dominates $x$). If $b'$ does not dominate 4, 
which implies $b'\neq x'$, the cycle $1x345$ is minimally dominated by 
a subset of $\{x',y',z',b'\}$. Since $x'$ only sees the vertex 4 of $1x345$
and is the only neighbor of $4$, we reach 
a contradiction to Lemma~\ref{lem:dominatingC5}. Thus $b'$ dominates 
both $x$ and 4. Finally $1xb'4yz'$ induces a $6$-cycle, a contradiction.

 \item In the last case, $X$ has three vertices $x',y',z'$ respectively 
 joined to $\{3,5\}$, $\{1,4\}$ and $\{2,4\}$. To get a {\it doubly extended $C_5$}, we just need to prove that there exists a vertex $w'$ in $N_{\ell -1}$ which sees $y'$ and $z'$ but not $x'$. This is exactly the same argument
ending the proof of Theorem~\ref{C5}.
\end{enumerate}
 
\end{proof}

Now we prove the last theorem that will enable us in the next sections to focus on $C_{5}$-free graphs.

\begin{theorem} \label{excludingC5}
If the chromatic number of trinity graphs not inducing doubly extended-$C_5$ is bounded,
then the chromatic number of trinity graphs is also bounded.
\end{theorem}

\begin{proof}
Assume that trinity-free doubly-extended $C_{5}$-free graphs have bounded chromatic number, and let $G$ be a connected trinity-free graphs with large chromatic number and fix some root $r$. According to Theorem~\ref{shadow}, we have a shadow or an antishadow. Again, as noted in the proof of Theorem~\ref{C5}, the antishadow case follows from Lemma~\ref{lem:antishadowC5}, so we can assume that there is a shadow. In particular one can find a large chromatic number subset $B$ in some $N_{\ell+1}$ which is dominated 
by some stable set $S$ in $N_{\ell}$. 

Since we assume that trinity-free doubly-extended $C_5$-free graphs have bounded chromatic number, Theorem~\ref{incrust} implies that we can find a TCP inside $B$ such that each block contains a doubly extended $C_5$. In particular, we can find a doubly extended 5-cycle $C$ in $B$ with $V(C)=\{1,2,3,4,5,x,y,z,w,x',y',z',w'\}$
 which is minimally dominated by a stable set $X'$ in $S$
such that $X'$ is rich. 
Let $X$ included in $X'$ minimally dominating 
$\{1,2,3,4,5\}$. By Lemma~\ref{lem:dominatingC5}, $X$ has three vertices, each joined 
to exactly two vertices of $\{1,2,3,4,5\}$. The only case not leading 
to a contradiction $\{1,2,3,4,5,x,y,z,w\}$ (just like in the proof of Theorem~\ref{extendedC5})
is when $X$ has three vertices $x'',y'',z''$ respectively 
 joined to $\{3,5\}$, $\{1,4\}$ and $\{2,4\}$. But then, as in the proof of 
Theorem~\ref{extendedC5}, we reach a contradiction with $\{1,2,3,4,5,x',y',z',w'\}$.
 
\end{proof}

\section{Excluding shadows in $C_5$-free trinity graphs}

We now show that we cannot find a shadow in a $C_5$-free trinity graph.
In this section, we assume for contradiction that $G$ is a $C_5$-free trinity graph with a shadow $B$ in $N_{\ell+1}$ 
which is dominated by a stable set $S$ in $N_{\ell}$. We can moreover assume 
that $S$ is rich in a subset $Y$ of $N_{\ell+1}$ which is independent of $B$.
To see this, just grow a large partial TCP $T$ starting from a vertex of $S$ and with all its other vertices in $B$. Here partial refers to the fact 
that $T$ ends at some vertex $v$ of $B$ which sees
some large chromatic subgraph $B'$ of $B$ such that $B'$ is independent 
from $T\setminus v$. Now the set of neighbors $S'$ of $B'$ in $S$ is rich 
in $T$, and we thus have our shadow with the aforementioned property 
(still using $B,S$ instead of $B',S'$). Every pair of vertices $x,y$ of 
$S$ has then a type $t_{xy}$.

\begin{lemma}\label{lempierre}
Let $123$ be an induced path in $B$ having respective private neighbors $x,y,z$ in $S$.
Then $y$ cuts $xz$,  the types of $xy$ and $yz$ are the same, and
$U_{xz}$ has trinity 0 (and is not seen by $y$).
\end{lemma}

\begin{proof}
First observe that $\{x,y,z\}$ have private neighbors in $N_{\ell-1}$ to
avoid $C_5$ and $C_6$.
In particular $U_{xy}$, $U_{yz}$ and $U_{xz}$ are not respectively
seen by ${z,x,y}$. Since $u_{xy}$ and $u_{yz}$
are non zero to avoid $U_{xy}21$ and $U_{yz}32$, we have
$u_{xy}=t_{xy}$ and $u_{yz}=t_{yz}$. We denote by $P_{xy}$
the path $x12y$ and by $P_{yz}$ the path $y23z$.

Assume for contradiction that $y$ does not cut $xz$. We can form an induced cycle of 
length $u_{yz}+t_{zx}+p_{xy}=u_{yz}+t_{zx}$. Since $u_{yz}$ and 
$t_{zx}$ have trinity $1$ or $2$, they must have the same trinity.
We can also form an induced cycle of 
length $p_{yz}+t_{zx}+u_{xy}=u_{xy}+t_{zx}$, giving $u_{xy}=t_{zx}$.
Thus $xy$ and $yz$ have the same type. 
But then there is a cycle of trinity $u_{xy}+t_{yz}+t_{zx}=0$, which
gives that $x$ cuts $yz$. Thus $z$ does not cut $xy$, and we reach 
a contradiction using an induced cycle of trinity $t_{xy}+u_{yz}+t_{zx}=0$.

So $y$ cuts $xz$, and in particular $z$ does not cut $xy$ and 
$x$ does not cut $yz$. There are induced cycles of trinity 
$t_{xy}+t_{yz}+u_{zx}$, $t_{xy}+p_{yz}+u_{zx}=t_{xy}+u_{zx}$ and 
$p_{xy}+t_{yz}+u_{zx}=t_{yz}+u_{zx}$. Assume for contradiction 
that $u_{zx}\neq 0$. Since $t_{xy}$ and $t_{yz}$ are also non zero,
we must have $t_{xy}=t_{yz}=u_{zx}$, giving a contradiction.
Thus $u_{zx}=0$, and therefore $t_{xy}=t_{yz}$.

\end{proof}

\begin{corollary} \label{P4private}
An induced path $1234$ in the shadow $B$ cannot have private 
neighbors $x,y,z,t$ in $S$.
\end{corollary}

\begin{proof}
Since $y$ cuts $xz$ and $z$ cuts $yt$, 
$y$ does not cut $zt$ and $z$ does not cut $xy$.
Assume for contradiction that $x$ cuts $yt$ or $zt$. Consider 
an induced path $T_{yt}$ (of trinity $t_{yt}=1$) in a typical block. 
It is seen 
by $z$, and then by $x$. But then, since $y$ also cuts $xz$,
this is only possible if $y$ has a common neighbor on
$T_{yt}$ with $x$ or $z$, giving a $C_5$ and a contradiction.

We then have an induced cycle of trinity $t_{xy}+t_{yz}+t_{zt}+u_{tx}$.
Since $t_{xy}=t_{yz}=t_{zt}$, we have $u_{tx}\neq 0$, and therefore
$u_{tx}=t_{tx}=2$ because of the path $x1234t$. We reach a contradiction
since there exists a cycle with trinity $p_{xy}+t_{yt}+u_{tx}=0$.
\end{proof}

\begin{corollary} \label{P5triple}
A vertex $y$ of $S$ cannot see $1,3,5$ in some induced path $12345$ in the shadow $B$.
\end{corollary}

\begin{proof}
To avoid $C_6$, the vertices $2,4$ must have private neighbors $x,z$
in $S$.
Applying Lemma~\ref{lempierre} to $xyz$ and $234$,
we have that $U_{xz}$ has trinity 0 and does not see $y$. Extend it 
with $z45y12x$ to get a trinity 0 cycle, and a contradiction.
\end{proof}

\begin{corollary}\label{P5three}
An induced path $12345$ in the shadow $B$ cannot have
neighbors $x,y,z$ in $S$ linked respectively to $\{1,3\}$, $\{2,4\}$
and $5$.
\end{corollary}

\begin{proof}
Since $3,4,5$ have private neighbors $x,y,z$, by Lemma~\ref{lempierre}
we have $u_{xz}=0$. But then, we can close a trinity 0 cycle
using $z54y21x$, a contradiction.
\end{proof}

\begin{corollary} \label{P7two}
An induced path $1234567$ in the shadow $B$ cannot have
neighbors $x,y$ in $S$ linked respectively to $\{1,3\}$ and $\{2,4\}$.
\end{corollary}

\begin{proof}
Let $z\in S$ be a neighbor of $5$. By Corollary~\ref{P5three},
$z$ must see 3. By the same argument, there is a neighbor 
$t\in S$ of 6 which sees 4, and then a neighbor $u\in S$ 
of $5,7$. This contradicts Corollary~\ref{P5three} applied
to $12345$ and $x,y,u$.
\end{proof}

\begin{theorem} \label{noshadowC5}
There is no shadow in a $C_5$-free trinity graph.
\end{theorem}

\begin{proof}
Since
$B$ has large chromatic number, it contains an induced path 
$0123456789$. By Corollary~\ref{P4private}, the vertices $0,2$ have a common 
neighbor in $S$ or $1,3$ have a common neighbor in $S$. 
Without loss of generality, let us assume that $x$ sees $1,3$.
By Corollary~\ref{P7two} there is $y,z\in S$ respectively 
seeing $2$ and $4$ on $1234$. Let $t\in S$ be a neighbor of $5$.
To avoid the private neighbors $y,x,z,t$ of $2345$ (a contradiction to Corollary~\ref{P4private}), the vertex
$t$ sees 3. By Corollary~\ref{P7two} applied to $3456789$,
we cannot have that $z$ sees $6$. Let $u \in S$ be a neighbor of $6$, and $v$ of $7$.
Note that $t$ does not see 7 by Corollary~\ref{P5triple}.
To avoid the private neighbors $z,t,u,v$ of $4567$ and a contradiction to Corollary~\ref{P4private}, the vertex 
$v$ sees 5. We get four private neighbors
$x,z,v,u$ of $3456$, a contradiction to Corollary~\ref{P4private}.
\end{proof}

\section{Trinity graphs have bounded chromatic number}

We now know that the original problem is equivalent 
to proving that antishadows do not exist in $C_5$-free case. 
Let us derive some results about antishadows in $C_5$-free
trinity graphs.

\begin{lemma}\label{antiP4}
Let $G$ be a $C_5$-free trinity graph with some antishadow $X_0,X_1,X_2,X_3,X_4$
where $X_0$ belongs to $N_{\ell}$. If $1234$ is an induced path in $X_0$
and $1,3$ have private neighbors $x,z$ in $X_1$, then $2,4$ do not have
a common neighbor in $X_1$.
\end{lemma}

\begin{proof}
Assume for contradiction that $2,4$ have a neighbor $y\in X_{1}$.
Let $1'$ be a neighbor of $1$ in $N_{\ell-1}$ and $3'$ be  a neighbor of 
$3$ in $N_{\ell-1}$ (we can have $1'=3'$, otherwise we assume that
$1',3'$ are private neighbors of $1,3$). The trinity $u_{1'3'}$ 
of a path $U_{1'3'}$ cannot be 2 since we could close a cycle
by $1'1233'$.
If $u_{1'3'}$  is 1, we have the two paths $x11'U_{1'3'}3'3z$ and $x123z$ 
with respective trinities 2 and 1, a contradiction to Fact~\ref{fact:richtype} in $X_{2}$. Thus $u_{1'3'}=0$.
We consider two (distinct) neighbors $x'$ and $y'$ of $x$ and $y$
in $X_{2}$. Note that the two paths $x'x11'U_{1'3'}3'34yy'$ and $x'x12yy'$ 
have respective trinities 1 and 2, a contradiction to Fact~\ref{fact:richtype} in $X_{3}$. 
\end{proof}

\begin{corollary} \label{antiP5triple}
A vertex $y$ in $X_1$ cannot see $1,3,5$ in some induced path $12345$ of $X_0$.
\end{corollary}

\begin{proof}
Applying Lemma~\ref{antiP4}, we deduce that there exists $x$
in $X_1$ dominating $2,4$. This gives the 6-cycle $y54x21$, and a contradiction.
\end{proof}

\begin{corollary}\label{nottwo}
If $1234567$ is an induced path (or cycle) in $X_0$ then
$1,3$ do not have common neighbors in $X_1$.
\end{corollary}

\begin{proof}
If $x\in X_1$ dominates $1,3$, by Lemma~\ref{antiP4}, there 
is $y\in X_1$ dominating $2,4$, and then $z\in X_1$ dominating $3,5$
and then $t\in X_1$ dominating $4,6$, and finally $u\in X_1$ dominating $5,7$.
By Corollary~\ref{antiP5triple}, $x$ does not see 5 and 
$u$ does not see 3. But then $3456$ and $x,t,u$ contradict
Lemma~\ref{antiP4}.
\end{proof}

\begin{theorem}\label{final}
Trinity graphs have bounded chromatic number.
\end{theorem}

\begin{proof}
Assume for contradiction that this is not the case, so by
Theorems~\ref{C5},~\ref{extendedC5},~\ref{excludingC5}, we can assume that there 
exists a $C_5$-free trinity graph $G$ with 
large chromatic number. By Theorem~\ref{noshadowC5},
we can assume that $G$ has an antishadow $X_0,X_1,X_2,X_3,X_4$.
Let $C$ be an induced odd cycle in $X_0$ (thus with at least 
7 vertices). By Corollary~\ref{nottwo}, $i,i+2$ on $C$ do
not have a common neighbor in $X_1$. Assume 
for contradiction that $C$ is minimally dominated by $|C|$
private vertices in $X_1$. If $C$ has trinity 2, consider (private) 
neighbors $x,y$ of $1,3$ in 
$X_1$. In $C$, $x,y$ can be joined by induced paths of length 1 and 2,
a contradiction to Fact~\ref{fact:richtype} in $X_2$. If $C$ has trinity 1, 
consider (private) neighbors $x,y$ of $1,4$ in 
$X_1$. If $x,y$ have respective private neighbors $x',y'$ in 
$X_2$, using $C$ and $x,y$, the vertices $x',y'$ 
can be joined by induced paths of trinities 1 and 2,
a contradiction to Fact~\ref{fact:richtype} in $X_3$. So $x,y$ do 
not have private neighbors 
in $X_2$, so we have $N(x)\cap X_2\subseteq N(y)\cap X_2$ or
$N(y)\cap X_2\subseteq N(x)\cap X_2$. We denote respectively 
these two cases by $1\rightarrow 4$ or $4\rightarrow 1$.
More generally we always have $i\rightarrow i+3$ or $i+3\rightarrow i$.
If $C$ has more than seven vertices, we conclude by
considering $z,t$ the respective (private) neighbors of $5,8$. There is
a vertex $z'$ in $X_2$ joined to $z,t$. Note that $x'\neq z'$, that 
the trinity of $x'y45zz'$ is 2, and that the trinity of $z't89\dots 1xx'$ 
is 1. This is a contradiction to Fact~\ref{fact:richtype} in $X_3$. Now if $C$ has 7 vertices,
observe that there exists $i$ such that $i\rightarrow i+3\rightarrow i+6$
(or equivalently $i+6\rightarrow i+3\rightarrow i$),
thus there exists a vertex in $X_2$ which is a neighbor of the
private neighbors of $i,i+3,i+6$, hence inducing a $C_5$ 
using $i,i+6$.

Thus there is a vertex $y$ in $X_1$
dominating two vertices of $C$, for instance $1$ and $k$. Observe that
since $k>3$, we have $k>5$ to avoid $C_5$ and $C_6$. We assume that $1,k$ are
chosen as close as possible among all possible choices of $y$. Note that $k$
does not have trinity 2 since $1\dots ky$ is an induced cycle.
If $k$ has trinity 0, take $w$ a neighbor of 3 in $X_1$,
and consider the paths $y123w$ and $w34\dots ky$
with trinities 1 and 2, a contradiction to Fact~\ref{fact:richtype}.
So $k$ has trinity 1.
Let $x,z,t,u\in X_1$ be the respective neighbors 
of $0,2,k-1,k+1$. All vertices $x,y,z,t,u$ are distinct to avoid $C_6$,
and they form a stable set since $X_1$ has no edges. If $x,t$
have private neighbors $x',t'$ in $X_2$, we obtain two paths
$x'x01yk(k-1)tt'$ and $x'x01\dots(k-1)tt'$ with trinities 2 and 1,
a contradiction to Fact~\ref{fact:richtype}. So there is $x'$ in $X_2$ neighbor of 
$x,t$. Similarly, there is $z'$ in $X_2$ neighbor of 
$z,u$. Note that $x'\neq z'$ to avoid the 6-cycle $x'x012z$. 
If $x'z'$ is an edge, 
we have the induced 9-cycle $x'z'u(k+1)ky10x$. 
Finally, if $x'z'$ is not an edge, we have the 12-cycle
$x'x012zz'u(k+1)k(k-1)t$, and this is our final contradiction.
\end{proof}

\bibliographystyle{plain}
 
\end{document}